\documentclass[%
 reprint,
 amsmath,amssymb,
 aps,
]{revtex4-2}

\usepackage{float}
\usepackage{graphicx}
\usepackage{dcolumn}
\usepackage{bm}
\usepackage{subcaption}

\usepackage{amsmath}
\usepackage{amsthm}

\newtheorem{theorem}{Theorem}
\newtheorem{lemma}[theorem]{Lemma} 

\usepackage{hyperref}

\begin{document}


\title{Challenges in Barren Plateau Mitigation with\\
Dynamic Parameterized Quantum Circuits}

\author{Sumeet Shirgure}
 \email{sumeetparameshwar.shirgure@ucf.edu}
\author{Efekan K\"okc\"u}
 \email{efekan.kokcu@ucf.edu}
\author{Siyuan Niu}%
 \email{siyuan.niu@ucf.edu}
 \affiliation{Department of Electrical and Computer Engineering,
 University of Central Florida.}

\date{\today}

\begin{abstract}
Variational quantum algorithms (VQAs) are a promising paradigm for quantum
advantage, yet their trainability is severely hampered by \emph{barren plateaus} (BPs).
Several recent works have proposed \emph{dynamic parameterized quantum circuits} (DPQCs), which interleave unitary layers with parameterized CPTP maps, such as engineered dissipation, feedforward gadgets, and periodic resets, as a possible strategy for mitigating BPs.
We unify this class of circuits into a formalization for DPQCs.
We identify constraints on the nature and the structure
of DPQCs if they are to prevent a significant number of parameters
from becoming untrainable.
Using purification and Pauli-path analysis, we further identify a mechanism by which the cost function can remain anti-concentrated even when many parameters remain untrainable.
Our analysis reveals ways to design DPQCs that do not have an exponentially concentrated
cost function, and our results suggest that BP mitigation via DPQCs is at least
as hard as designing BP-free unitaries.
\end{abstract}

\maketitle

\section{Introduction}\label{sec:introduction}

Variational quantum algorithms (VQAs) \cite{cerezo2021variational} are among the most
actively pursued strategies for quantum optimization.
They provide a flexible, general-purpose framework: the solution to a problem is encoded
in the minimum of a cost function, which is typically the expectation value
of a problem observable measured on the quantum state that a parameterized quantum circuit
(PQC), also called an \emph{ansatz}, prepares from a fixed initial state.
A classical optimizer then searches for the parameter vector that minimizes the cost,
and the resulting state, or its observed eigenvalue, provides the answer.
This way information is fed back and forth between a classical optimizer
that tunes the parameters and a quantum computer that computes
expectation values for a given set of parameters.
This hybrid quantum--classical recipe underlies many prominent algorithms, including
the Quantum Approximate Optimization Algorithm (QAOA) \cite{farhi2014quantum} and the
Variational Quantum Eigensolver (VQE)
\cite{peruzzo2014variational,tilly2022variational}.

The central obstacle to scaling VQAs is the \emph{barren plateau} (BP) phenomenon
\cite{mcclean2018barren, larocca2025barren}.
In a barren plateau regime the variance of the cost-function gradient with respect to a
parameter shrinks exponentially in the number of qubits $n$; such a
parameter is said to be \emph{untrainable}, because the number of circuit evaluations
needed to resolve its gradient from statistical noise grows exponentially, rendering
classical optimization intractable.
Barren plateaus have several distinct origins \cite{larocca2025barren}.
They can be induced by an overly expressive circuit \cite{mcclean2018barren}, by the choice
of input state and measurement operator \cite{ragone2024lie,cerezo2021cost}, or by
noise \cite{wang2021noise}.
Ultimately, they reflect the curse of dimensionality of an exponentially large Hilbert space.

A closely related quantity is the spread of the cost function itself.
It was shown in \cite{arrasmith2022equivalence} that if \emph{every} partial gradient
in a PQC is exponentially suppressed, then the cost function also concentrates
exponentially tightly around its mean.
(We make these definitions and concepts mathematically precise in section ~\ref{sec:preliminaries}.)
The converse, however, is weaker: showing that the cost function does \emph{not}
concentrate, for instance by lower-bounding the cost function variance across the parameter
space \cite{deshpande2024dynamic,Crognaletti2024}, only guarantees that \emph{some} gradients
are non-negligible on average, not that all of them are.
In this work we investigate which gradients are concentrated and which are not.
Cost-function anti-concentration is therefore a \emph{necessary} but not a \emph{sufficient}
condition for trainability, a distinction that will be central to our analysis.

Many strategies have been proposed to avoid or mitigate BPs, and most of them are focused on \emph{unitary} VQAs
\cite{larocca2025barren}. These strategies include: (1) using \emph{shallow circuits}, since deep circuits can approximate $2$-designs and hence are very expressive and have exponentially
small gradients \cite{mcclean2018barren}; (2) exploiting circuit or observable symmetries to obtain \emph{small dynamical Lie algebras} and thereby reduce effective expressivity~\cite{ragone2024lie}; (3) adopt \emph{local cost functions} , which generally exhibit more favorable trainability than global cost functions~\cite{cerezo2021cost}; (4) employing \emph{informed initialization strategies}, such as initializing the optimizer with carefully chosen parameters, e.g., small-angle parameters, rather than randomly selected values~\cite{zhang2022escaping,wang2024trainability,sack2022avoiding,park2024hardware}, or using warm-start methods that initialize the optimization near a good guess to the solution,
obtained by other methods ~\cite{puig2025,mhiri2025,zunkovic2026}.

Beyond these unitary techniques, a separate line of work mitigates BPs with explicitly
\emph{non-unitary} operations. These include: (1)\textit{Engineered dissipation:} for shallow state-preparation circuits, tailored
dissipation \cite{sannia2024engineered,cichy2024perturbative} can effectively turn a
global observable into a local one, mitigating BPs. (2) \textit{Feedforward gadgets:} measuring ancilla qubits and conditionally applying gates
based on the outcome \cite{deshpande2024dynamic} can produce cost-function
anti-concentration. (3) \textit{Periodic resets:} probabilistically resetting qubits can anti-concentrate
the cost function and protect the gradients of parameters in the final $O(\log(n))$ layers of the
ansatz \cite{zapusek2025scaling}, where $n$ is the number of qubits. (4) \textit{Mutating the unitary with non-unitary gadgets:} entangling the system with fresh
ancillas and then discarding them can restore trainability of parameters that follow the
gadget, provided the gadget is placed before the last $O(\log(n))$ layers
\cite{chen2025taming}.

The unitary approaches have been explored extensively, and are well understood 
\cite{larocca2025barren, ragone2024lie}.
The study of non-unitary approaches however has been disparate and
is not as explored in depth as their unitary counterparts.
This work aims to close that gap by providing a unifying perspective
on non-unitary approaches to BP mitigation, their potential, and their limitations.
For simplicity, the non-unitary channels explored in this work are assumed to be precisely engineered and error-free, since noise can render VQAs classically simulable under a range of conditions~\cite{gonzalez2025pauli,mele2026noise,cerezo2025does}.

Our contributions are as follows :
We unify the non-unitary proposals above into a single formalism, that we call
\emph{dynamic parameterized quantum circuits} (DPQCs),
and study the challenges in BP mitigation within it.
We prove two results establishing conditions on the properties and placement of non-unitary elements that are necessary for a substantial fraction of the parameters in a DPQC to remain trainable. We then develop a method based on Pauli-path propagation to determine whether the inserted non-unitary gadgets can induce cost-function anti-concentration. The method is computationally efficient when the number of Pauli terms explored by the algorithm remains tractable. Our construction demonstrates that a DPQC can exhibit cost-function anti-concentration while still containing a large number of untrainable parameters, showing that anti-concentration alone is insufficient to guarantee trainability across the entire circuit. We illustrate our theoretical results numerically using two practically relevant DPQCs, examining both their trainability and optimization landscapes. Finally, we argue that designing a BP-free DPQC is at least as difficult as constructing a BP-free unitary ansatz from first principles.

\section{Results}\label{sec:results}

\subsection{Preliminaries}
\label{sec:preliminaries}
In this section, we provide the formal statements that underlie the discussion above.

A VQA encodes a problem in the minimum of a cost function $C$.
While $C$ could be of various forms, and even possibly non-linear,
we will consider the following cost function: 
\begin{equation}
    C(\theta) = \mathrm{Tr}[H \hat{U}(\theta) \rho_0 \hat{U}^\dagger(\theta)]
    \label{eqn:cost}
\end{equation}
where $\rho_0$ is the $n$-qubit initial state, $H$ is the observable of interest,
and $\hat{U}(\theta)$ is the ansatz. We consider an ansatz that consists of $L$ blocks :,
\begin{equation}
    \hat{U}(\theta) = \hat{U}_L(\theta_L)...\hat{U}_2(\theta_2)\hat{U}_1(\theta_1)
    \label{eqn:ansatz}
\end{equation}
and each block has the following form
\begin{equation}
    \hat{U}_j(\theta_j) = \prod_{k} e^{-i\theta_{j,k} P_{j,k} / 2} \hat{W}_{j,k}
    \label{eqn:unitary}
\end{equation}
where $\hat{W}_{j,k}$ are non-parametrized unitaries (such as Hadamard or CNOT gates),
and the generators $P_{j,k}$ are local with a bounded Pauli weight, and they square
to the identity i.e. $P_{j,k}^2 = \mathrm{I}$.
This makes $C$ periodic in every component of $\theta$, and the standard
parameter shift rules for evaluating gradients apply
\cite{mitarai2018quantum,wierichs2022general}.
Note that we denote the unitary operation with a hat : $\hat{U}(\theta)$,
while we denote the unitary channel generated by it
$U(\theta):\rho \rightarrow \hat{U}(\theta) \rho \hat{U}(\theta)^\dagger$
without the hat. We denote a parameterized channel acting on a Hermitian operator
$\rho$ with the following notation $U(\theta)[\rho]$,
where the round brackets show the parameter and the square bracket
the input to the channel.

The classical optimizer returns $\theta^* = \arg\min_\theta C(\theta)$, and the
corresponding unitary is used to generate the quantum states of interest.
The part of the process that iteratively keeps updating the parameters is called "training".
An \emph{untrainable} parameter $\theta_\mu$ is one whose partial gradient is exponentially
small in $n$ on average :
\begin{equation}
    \mathrm{Var}_\theta[\partial_{\theta_\mu} C(\theta)]\in O(b^{-n}), \quad b>1
    \label{eqn:bp}
\end{equation}
A useful companion statement, proved in \cite{arrasmith2022equivalence},
is that for cost functions of the kind given in Eq. \ref{eqn:cost}, when all
partial gradients are exponentially suppressed the cost itself concentrates about its mean,
\begin{equation}
  \Pr_\theta [ \left| C(\theta) - \mathbb{E}_\phi [ C(\phi) ] \right| > \delta ] \in O(b^{-n})/\delta^2
  \label{eqn:cost_conc}
\end{equation}
where the probability is taken with respect to the uniform distribution over the parameter
space.
As emphasized in section \ref{sec:introduction}, the implication runs only
one way: anti-concentration of $C$ is necessary for trainability but does
not by itself guarantee it.

\subsection{Formalism}
\label{sec:formalism}
In this section we provide some definitions needed to construct and
describe non-unitary ansatzes.

\begin{figure*}[t]
    \begin{subfigure}[t]{0.45\textwidth}
        \centering
        \includegraphics[height=1.25in]{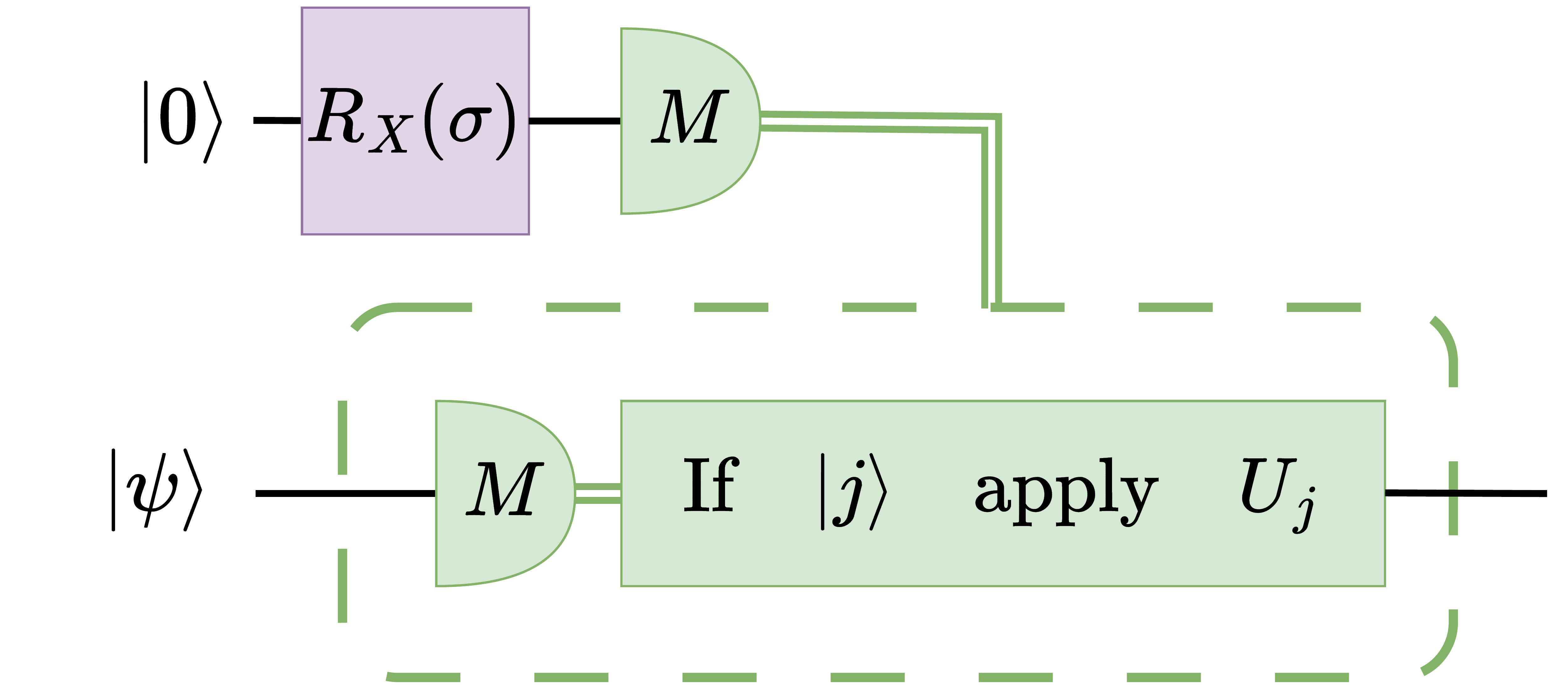}
        \caption{}
        \label{fig:feedforward_gadget}
    \end{subfigure}
    \hspace{10pt}
    \begin{subfigure}[t]{0.45\textwidth}
        \centering
        \includegraphics[height=1.25in]{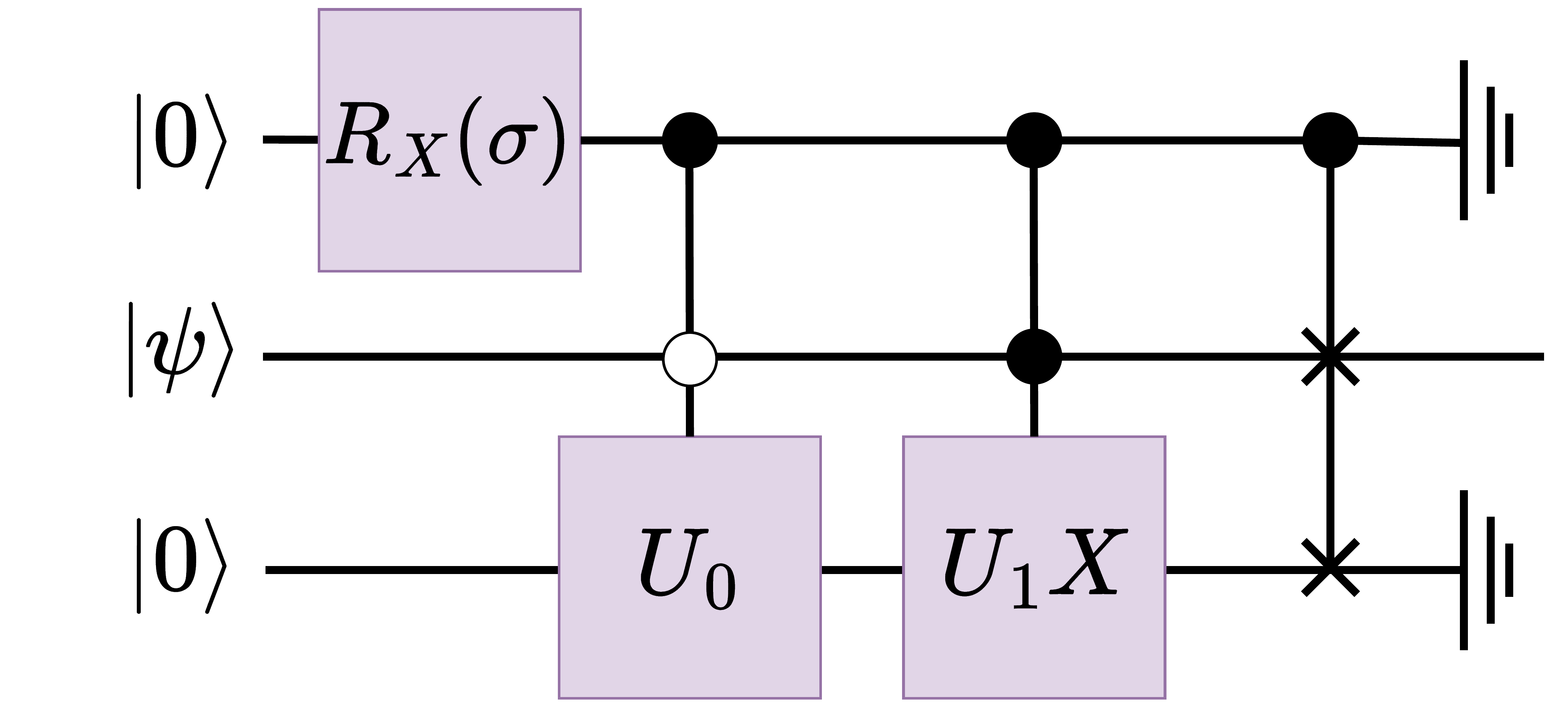}
        \caption{
        }
        \label{fig:purified_feedforward_gadget}
    \end{subfigure}
    \caption{
        An example of a non-unitary gadget.
        (\ref{fig:feedforward_gadget})
        The $\sigma$-parameterized ``feedforward gadget" from
        \cite{deshpande2024dynamic}.
        The sub-circuit in the dotted green box is applied conditionally
        on the measurement outcome of the ancilla at the top.
        We set $U_0=\mathrm{I},U_1=\frac{I-iX}{\sqrt{2}}$.
        (\ref{fig:purified_feedforward_gadget})
        The purified implementation of the gadget, using an additional ancilla
        at the bottom.
        The symbols at the right end of the ancilla lines are ``discard" operations.
        A probabilistic reset can be considered a specialization of this gadget
        with $U_0=U_1X=\mathrm{I}$ and reset probability $\sin^2(\sigma/2)$.
    }
    \label{fig:non_unitary_gadgets}
\end{figure*}

\textit{Parameterized dynamic gadgets.}
The basic non-unitary building block is defined as a \emph{parameterized dynamic gadget}:
a local CPTP map $\mathcal{E}_{j,k}(\sigma_{j,k})$ that acts on $O(1)$ qubits
and depends on a parameter $\sigma_{j,k}$.
Physically, the map is realized by bringing in a few ancilla qubits,
entangling them with the system qubits by a unitary, and then discarding the ancillas.
By Stinespring's dilation theorem \cite{stinespring_dilation,Nielsen2000} this
construction can in principle implement \emph{any} CPTP map;
however for practical purposes we restrict our attention to those that
require only $O(1)$ ancillas.
Fig.~\ref{fig:non_unitary_gadgets} shows a concrete example.

\textit{Dynamic parameterized quantum circuits (DPQCs).}
A DPQC $\mathcal{U}(\theta,\sigma)$ is a quantum channel obtained
from an ordinary PQC $U(\theta)$ (Eq.~\ref{eqn:ansatz}) by
inserting $M$ layers of non-unitary maps $\mathcal{E}_j(\sigma_j)$
between its unitary blocks (Eq. ~\ref{eqn:dpqc}),
\begin{widetext}

\begin{equation}
    \mathcal{U}(\theta,\sigma) = 
    \mathcal{E}_M(\sigma_M) \circ U_L(\theta_L) \circ
    ...\mathcal{E}_j(\sigma_j) \circ U_{v_{j}}(\theta_{v_{j}})...\circ U_{v_{j-1}+1}(\theta_{v_{j-1}+1})
    \circ \mathcal{E}_{j-1}(\sigma_{j-1})...\circ U_1(\theta_1)
    \label{eqn:dpqc}
\end{equation}

\end{widetext}
where each inserted layer is itself a composition of dynamic gadgets,
\begin{equation}
    \mathcal{E}_j(\sigma_j) =
    \mathcal{E}_{j,l_j}(\sigma_{j,l_j})\circ...\circ \mathcal{E}_{j,1}(\sigma_{j,1})
    \label{eqn:cptp_composition}
\end{equation}
The cost function is defined exactly as before,
now with respect to the modified channel,
\begin{equation}
    C(\theta, \sigma) = \mathrm{Tr}[H\ \mathcal{U}(\theta,\sigma)[\rho_0]]
    \label{eqn:dpqc_cost}
\end{equation}

\textit{Preserving expressivity.}
An unconstrained CPTP map can reduce the expressivity drastically,
since CPTP maps such as a reset channel can erase any quantum information stored.
In order to guarantee that $\mathcal{U}$ is at least as expressive as the unitary PQC
$U$, we require each gadget to reduce to the identity channel when its parameter is
equal to $0$, i.e. $\mathcal{E}_{j,k}(0)=\mathrm{I}$.
The parameter $\sigma$ then smoothly interpolates between the unitary
circuit $U$ (at $\sigma=0$) and the fully dynamic circuit $\mathcal{U}$,
thus guaranteeing that DPQC is at least as expressive as the PQC.

We note that existing works \cite{deshpande2024dynamic,chen2025taming}
also deal with gadgets that are capable of realizing the reset channel.
Therefore such expressivity arguments are only valid for small $\sigma$,
and for larger values the ansatz may potentially lose its expressivity.

\textit{Faithfulness.}
The inserted dynamic layers can also reshape the cost landscape and
destroy the encoding the original PQC carried.
We call a DPQC $\mathcal{U}$ \emph{faithful} to $U$ for a given
$(\rho_0, H)$ if changing the parameter slightly perturbs the cost $C$
\emph{uniformly} by at most an exponentially small amount at each point:
\begin{equation}
\begin{split}
  \left|C(\theta,\sigma) - C(\theta,0)\right| \in \mathcal{O}(c^{-n}), c>1\\
  \forall\theta_\mu \in [-\pi,\pi),\forall \sigma:\ ||\sigma|| < \epsilon \in
        \Theta \left( \frac{1}{poly(n)} \right)
\end{split}
  \label{eqn:faithful}
\end{equation}
This is a stringent requirement, but it is the natural one here: the original
cost is \emph{already} exponentially concentrated (Eq.~\ref{eqn:cost_conc}), so any change
larger than $O(c^{-n})$ could deform the landscape into one that no longer encodes
the original problem in a faithful manner.
This definition allows us to critically examine the tempting claim that a DPQC $\mathcal{U}$ can simultaneously remain faithful to the original problem, be as expressive as the unitary PQC $U$, and be trainable.
As noted earlier, the expressivity argument holds only for small $\sigma$;
but we will show that small $\sigma$ is precisely the regime in which
a faithful $\mathcal{U}$ resembles $U$ and therefore inherits its BPs.

\subsection{Number of untrainable parameters in DPQCs}
\label{sec:dpqc_trainability}

\begin{lemma}[Untrainability of faithful DPQCs]
    If $(U,H,\rho_0)$ exhibits a barren plateau in the
    sense of Eqs.~\ref{eqn:bp} and \ref{eqn:cost_conc}, then
    for a faithful DPQC $\quad \mathcal{U}(\theta,\sigma)$ constructed from
    $U(\theta)$ as defined in Eq. \ref{eqn:faithful},
    if $U$ has no trainable parameters with respect to $(H, \rho_0)$,
    then no $\theta$ parameter is trainable in $\mathcal{U}$.
    \label{lemma:faithful_dpqc}
\end{lemma}
Lemma~\ref{lemma:faithful_dpqc} (proved in Appendix \ref{sec:methods})
states that faithfulness and trainability cannot
be reconciled for free: a faithful DPQC inherits the untrainability of the unitary
it was built from. Therefore, the augmented cost function $C(\theta,\sigma)$ can overcome this limitation only if the new $\sigma$  directions encode information relevant to the solution of the original problem, rather than merely preserving the features of the original cost landscape.  Such a modification requires giving up
faithfulness.
Consequently, constructing DPQCs that are both expressive and trainable generally requires sacrificing faithfulness to the original unitary PQC.
The practical consequence is that effective DPQCs must be tailored, with an inductive bias
toward the problem at hand; a single ``one-size-fits-all" dynamic gadget is unlikely to
mitigate BPs across arbitrary VQAs.

A second limitation concerns \emph{how often} the gadgets are inserted.
If two consecutive non-unitary layers are separated by a long stretch of unitary blocks
$\hat{U}_{v_j}(\theta_{v_j})...\hat{U}_{v_{{j-1}}+1}(\theta_{v_{{j-1}}+1})$,
that entire stretch can become untrainable if it forms a 2-design.

\begin{lemma}[Untrainability of sparse DPQCs]
    If any sufficiently long sub-array of $U(\theta)$, namely
    $\hat{U}_{j+k-1}(\theta_{j+k-1})...\hat{U}_j(\theta_j),
    k\in\Omega(L)$ as defined in Eq. \ref{eqn:ansatz} forms a 2-design,
    then for a DPQC $\quad \mathcal{U}(\theta,\sigma)$ constructed from
    $U(\theta)$ as defined in Eq \ref{eqn:dpqc},
    if $v_j-v_{j-1}\in \Omega(L)$ then no parameter in the
    $U_{v_{j}}(\theta_{v_{j}})...U_{v_{j-1}+1}(\theta_{v_{j-1}+1})$
    part of the dynamic circuit is trainable.
    \label{lemma:sparse_dpqc}
\end{lemma}

Lemma~\ref{lemma:sparse_dpqc} rules out a popular design choice:
inserting only a constant number $M=O(1)$ of gadget layers to mitigate BPs.
This is considered for example in
\cite{deshpande2024dynamic,chen2025taming}.
It cannot prevent a significant fraction of the parameters from becoming untrainable.
Counting conservatively, if each layer $U_j(\theta_j)$ carries $\Omega(n)$ parameters,
this leaves potentially $\Omega(nL)$ parameters stranded on a plateau.
The key assumption, that any long sub-array of layers forms a $2$-design,
is mild in practice because most parameterized circuit families are
translation invariant: every sufficiently long window of layers has the
same Haar-random properties.
Translationally invariant families include deep random circuits
\cite{harrow2009random}, the hardware-efficient ansatz (HEA)
\cite{kandala2017hardware}, the QAOA ansatz \cite{farhi2014quantum},
and Trotterized variational versions of the Unitary Coupled Cluster ansatz
\cite{lee2018generalized,anand2022quantum}.
Of these, random circuits and HEA are known to form $2$-designs once deep enough
\cite{cerezo2021cost}, and Max-Cut QAOA is known to have an exponentially large
dimensional Lie algebra for most graphs
\cite{mao2025qaoa,kazi2025analyzing,kokcu2024classification}.

\subsection{Pauli path analysis}
\label{sec:pauli_path_analysis}

In this section, we provide a powerful tool to analyze the existing
forms of mitigating BPs in DPQCs.
In the same spirit as unifying the existing mitigation techniques under the DPQC
formalism, this tool is a unified analysis method that works on all said techniques.
It is based on \emph{Pauli-path analysis}
~\cite{bravyi2016improved,aharonov2023polynomial,gonzalez2025pauli,angrisani2025classically},
which works by expressing the quantum observables in the Pauli basis.
The idea is to evolve the observable in the Heisenberg picture and track it
as a linear combination of Pauli strings.

Since it is a popular design choice,
we focus on a single non-unitary layer $\mathcal{E}(\sigma)$ placed somewhere inside the
circuit, and ask how far from the final measurement it can sit.
Any unitary acting \emph{after} the dynamic layer can be absorbed into it, so without loss
of generality we write
\begin{equation}
    \mathcal{U}(\theta,\sigma)=\mathcal{E}(\sigma)\circ U(\theta)
    \label{eqn:simplified_dpqc}
\end{equation}

Note that the $\sigma$ and $\theta$ parameters here denote different things,
than what they did in Eq. \ref{eqn:dpqc}. We have "absorbed" all of the parameters
in the dynamic layer and the unitary coming after it into $\sigma$, and the
rest of the parameters are denoted by $\theta$.

To bring this into the Pauli-path framework we purify \cite{stinespring_dilation} the
non-unitary channel $\mathcal{E}(\sigma)$ into a unitary channel $V(\sigma)$ acting on the
enlarged Hilbert space $\mathcal{H}_s\otimes\mathcal{H}_a$, where $\mathcal{H}_s$ holds the
system qubits and $\mathcal{H}_a$ the ancillas.
Initializing the ancillas in $|0\rangle$ and discarding them at the end reproduces the
original DPQC from Eq. \ref{eqn:simplified_dpqc},
\begin{equation}
    C(\theta,\sigma) =
    \mathrm{Tr}[H\
    \mathrm{Tr}_a[V(\sigma)\circ U(\theta) [\rho_0\otimes |0\rangle\langle0|_a]]]
    \label{eqn:purified_dpqc}
\end{equation}

\begin{figure}[t]
    \centering
    \includegraphics[width=0.9\linewidth]{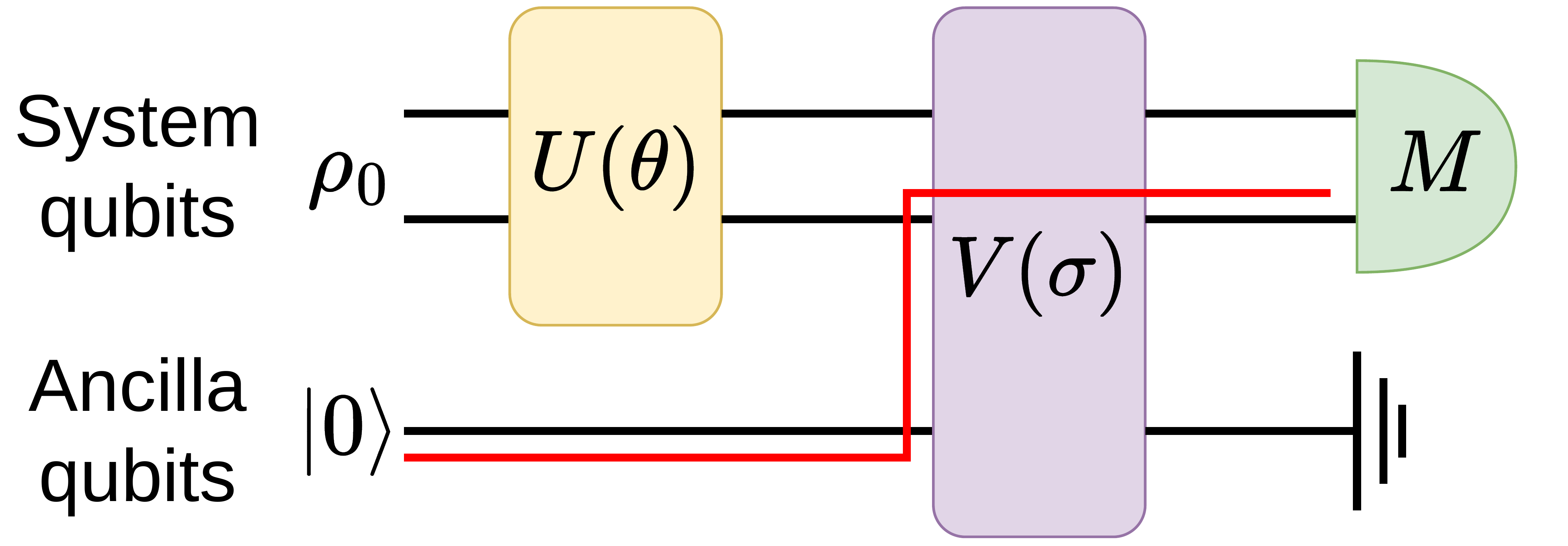}
    \caption{
    Schematic diagram of purification of a non-unitary channel
    as defined in Eqs. \ref{eqn:simplified_dpqc} and \ref{eqn:purified_dpqc}.
    The red line denotes Pauli strings in the linear sum representing
    the observable $H$, evolving into strings supported solely on $\mathcal{H}_a$.
    }
    \label{fig:purification}
\end{figure}

This setup is summarized in Fig.~\ref{fig:purification}.
The key observation is that, although $H$ starts out supported entirely on the system,
the Heisenberg evolution $\hat{V}(\sigma)^\dagger H_s \otimes I_a \hat{V}(\sigma)$
generates some Pauli terms that live \emph{entirely on the ancillas}.
These ancilla supported terms are special: their coefficients depend only on $\sigma$, and
because $U(\theta)$ acts as $U(\theta)_s\otimes I_a$, they are untouched (not scrambled) by
the preceding parameters $\theta$.

Every existing method for achieving cost anti-concentration works by engineering the
gadgets so that the sum of these ancilla supported terms has high variance,
while keeping $\hat{V}(\sigma)$ simple.
For example, $\hat{V}$ has local gates and has depth $O(\log(n))$
\cite{chen2025taming,zapusek2025scaling} or even depth $O(1)$ \cite{deshpande2024dynamic}.
However, in such settings, the $\hat{V}(\sigma)$ part of the circuit can
be analyzed efficiently with Pauli paths :

\begin{lemma}[Cost anti-concentration from ancilla supported Pauli terms]
    Let $\mathcal{U}(\theta,\sigma)$ be a purified
    DPQC as defined in Eq. \ref{eqn:simplified_dpqc}.
    Suppose that $U(\theta)$ as defined forms a 1-design over $\theta$.
    And let $\hat{V}(\sigma)$ be such that the Heisenberg evolution of
    $H_s\otimes I_a$ through $\hat{V}$ can be computed classically as a sum of
    polynomially many terms
    \begin{equation}
    \begin{split}
        \hat{V}(\sigma)^\dagger H \hat{V}(\sigma) &=\sum_{P=I_s\otimes P_a}f_P(\sigma) I_s\otimes P_a \\
        &+\sum_{Q=Q_s\otimes Q_a, Q_s\neq I_s}{f_Q(\sigma) Q_s\otimes Q_a}
    \end{split}
    \end{equation}
    Define the classically computable function
    $F(\sigma) = \sum_P f_P(\sigma) \mathrm{Tr}[P_a |0\rangle\langle0|_a]$.
    If we can design $\hat{V}(\sigma)$ such that $\mathrm{Var}_\sigma[F(\sigma)] \in
    \Omega(1/poly(n))$,
    then
    $\mathrm{Var}_{\theta,\sigma}[C(\theta,\sigma)]
        \geq \mathrm{Var}_\sigma[F(\sigma)] \in \Omega(1/poly(n))$.
    \label{lemma:pauli_path_lemma}
\end{lemma}

Lemma~\ref{lemma:pauli_path_lemma} (proved in Appendix \ref{sec:methods})
provides a concrete certificate: once the circuit preceding the
gadget layer is random enough to behave like a $1$-design, we can certify
anti-concentration in $C$ simply by designing gadgets in a way that lower-bounds
the variance of the classically analyzable function $F(\sigma)$.
Note that we make the distinction between simulating a parameterized quantum circuit
with some parameter instantiation, and deriving analytical expressions
for some Pauli terms in Lemma \ref{lemma:pauli_path_lemma}.
The latter is only feasible when $\hat{V}(\sigma)$ is simple enough to analyze symbolically,
but it carries a real advantage: the gadgets can be tailored to the specific observable $H$.
To carry out such calculations we provide a software tool ``sympauli", a symbolic
``Pauli-Heisenberg evolution" engine that returns closed-form analytic expressions
$f_P(\sigma)$ given a PQC and an observable with polynomially many Pauli terms.
However, in general for a log-depth circuit with local gates the number of gates
in a local observable's light cone can be $O(\log(n)^2)$, so the complexity
of symbolic coefficients in ``sympauli" can potentially go up to $n^{\Theta(log(n))}$.
Thus if analytical expressions for coefficients of Pauli terms are intractable,
we can always fall back to numerical Monte-Carlo integration with
randomly sampled parameter instantiations for evaluating the variance
$\mathrm{Var}_\sigma[F(\sigma)]$.

The mechanism in Lemma \ref{lemma:pauli_path_lemma} also explains an asymmetry
seen in earlier works \cite{deshpande2024dynamic,zapusek2025scaling,chen2025taming}:
dynamic gadgets can restore
trainability for the parameters that come \emph{after} them ($\sigma$),
and more importantly, cannot guarantee it for the parameters that come
\emph{before} ($\theta$).
The ancilla supported terms supply variance to the cost, yet that effect is
controlled by $\sigma$ alone and can potentially never reach the upstream
$\theta$ parameters.
We note the similarity, and possibly a connection, between this phenomenon and the one
observed in \cite{mele2026noise}, namely that noisy deep circuits behave like
shallow circuits with only the last layers contributing significantly to the
observable expectation value.
Standard noise channels such as depolarizing and
amplitude damping are CPTP maps with fixed parameters, so our analysis might be
applied to noisy circuits, and in the other direction, the results of
\cite{mele2026noise} might be extended to DPQCs as well.
It also sharpens a question about the ``BP-free" noisy circuits studied in
\cite{Crognaletti2024}, namely the possibility of cost anti-concentration
as shown in their work, which, as we have shown, can occur even in the
presence of significantly many untrainable parameters.

Moreover if we let $\hat{V}(\sigma)$ get too complex, by becoming too deep for example,
it becomes harder to keep
$\mathrm{Var}_\sigma[F(\sigma)] \in \Omega(1/poly(n))$, and we risk even the
ancilla ($\mathcal{H}_a$) supported Pauli terms $P_a$ becoming too insignificant after getting
scrambled by $V$, so that $\sigma$ becomes untrainable as well.
Our numerical results in Fig.~\ref{fig:variance_lower_bound_qaoa} support this claim, and it
is consistent with Lemma~\ref{lemma:sparse_dpqc} applied to the sub-array \emph{after} the
dynamic layer as it approaches a $2$-design.
For deep $V$ it also becomes harder to certify analytically that a given gadget induces
anti-concentration at all.

This suggests a natural principle for designing BP-free DPQCs :
insert \emph{multiple} layers of dynamic gadgets, densely enough that no
sub-array can form a $2$-design (cf.\ Lemma~\ref{lemma:sparse_dpqc}),
and use Pauli path analysis to build the layers up step by step.
We note that using it for layers far from the end becomes
challenging as the number of Pauli strings may explode exponentially,
and this can be seen as the main challenge of BP-free ansatz design.
The Pauli-path algorithm of \cite{angrisani2025classically} may
be repurposed for deep stacks of \emph{locally scrambling} unitary layers,
which makes the number of Pauli terms tractable by the means of truncation.
Furthermore, note that each subarray of unitary layers in Eq. \ref{eqn:dpqc}
can effectively be ``absorbed" into the adjacent CPTP map.
Because the whole construction lives in the
purified picture, we see that designing such a BP-free dynamic circuit ultimately
reduces to designing a BP-free \emph{unitary} circuit where some subset of qubits
are initialized to $|0\rangle$.

In Sec.~\ref{sec:numerical_analysis} we confirm numerically,
for applications of interest like VQE and QAOA, that the $\theta$
parameters of an ansatz in the form of Eq. \ref{eqn:simplified_dpqc}
indeed remain untrainable for deep $U(\theta)$.

\subsection{Numerical analysis}
\label{sec:numerical_analysis}

We study two DPQCs of the form given in Eq. \ref{eqn:simplified_dpqc} and show
numerically that a circuit can have an anti-concentrated cost,
$\mathrm{Var}_{\theta,\sigma}[C(\theta, \sigma)]\in \Omega(1/poly(n))$, while its
$\theta$ parameters remain untrainable.
In other words, essentially all of the variation in $C$ comes from the $\sigma$ directions,
with only an exponentially small contribution from the $\theta$ directions.
We summarize the untrainability of an entire subset parameters at once using a
single statistic. We first rewrite the cost using the adjoint channel
$\mathcal{E}^\dagger(\sigma)$.
\begin{equation}
    C(\theta, \sigma)
        = \mathrm{Tr}\left[
        \mathcal{E}^\dagger(\sigma)[H] U(\theta)[\rho_0] \right]
    \label{eqn:adjoint_representation}
\end{equation}
Next, using the law of total variance together with the periodicity of $C$ in $\theta$, we
decompose the variance of $\partial_{\theta_\mu}C$ over the entire parameter space
$(\theta, \sigma)$ into an average of variances taken over affine slices at fixed $\sigma$,
(see Appendix~\ref{sec:methods} for derivation)
\begin{equation}
    \mathrm{Var}_{\theta,\sigma}[\partial_{\theta_\mu}C]
        = \mathbb{E}_\sigma\left[\mathrm{Var}_\theta\left[
            \partial_{\theta_\mu} C|\sigma\right]\right]
    \label{eqn:law_of_total_variance}
\end{equation}

Together, Eqs.~\ref{eqn:adjoint_representation} and \ref{eqn:law_of_total_variance} give a
practical diagnostic for BPs in affine subspaces.
For each randomly sampled $\sigma$, we estimate the sample variance of the cost differences
$C(\theta, \sigma)-C(\theta',\sigma)$ over pairs of uniformly randomly sampled
points $(\theta,\theta')$, that is
$D(\sigma)=\mathrm{Var}_{\theta,\theta'}\left[
        C(\theta, \sigma)-C(\theta',\sigma)|\sigma\right]$,
and then average over $\sigma$.
Since $\mathcal{E}^\dagger(\sigma)(H)$ is another observable,
and $U(\theta)$ is unitary, by the same argument as in
\cite{arrasmith2022equivalence}, an exponentially small value of
the average of this statistic $\mathbb{E}_\sigma[D(\sigma)]$ implies that the
partial gradient $\partial_{\theta_\mu}$ for every $\theta_\mu$ is also concentrated.
This sidesteps costly explicit gradient evaluations and, importantly, avoids automatic
differentiation libraries that only work on the full $\mathcal{H}_s\otimes\mathcal{H}_a$
space and therefore cannot cope with a large number of gadgets/ancillas.
We instead use plain state-vector simulation with Qiskit's AerSimulator \cite{qiskit2024},
reusing a single ancilla qubit by serializing all of the gadgets and resetting the
ancilla to $|0\rangle$ between gadget instantiations.

We use these techniques to numerically analyze two practical DPQC applications,
ground state preparation with VQE and solving Max-Cut with QAOA.
We test several dynamic gadgets across these applications, and observe
cost function anti-concentration in some cases but not in others.
Finally, we probe how far before the final measurement the dynamic layer can be
placed in QAOA before the cost function variance lower bound
drops to exponentially small values.

\begin{figure*}[t]
    \begin{subfigure}[t]{0.45\textwidth}
        \centering
        \includegraphics[height=2.25in]{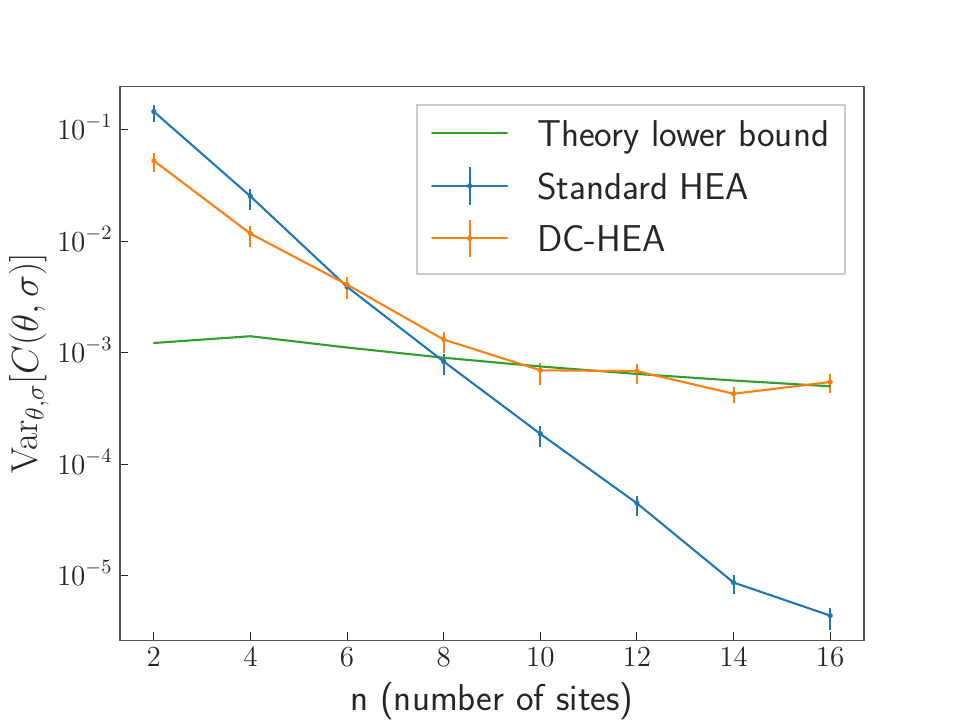}
        \caption{}
        \label{fig:cost_anticoncentration_vqe}
    \end{subfigure}
    \hspace{10pt}
    \begin{subfigure}[t]{0.45\textwidth}
        \centering
        \includegraphics[height=2.25in]{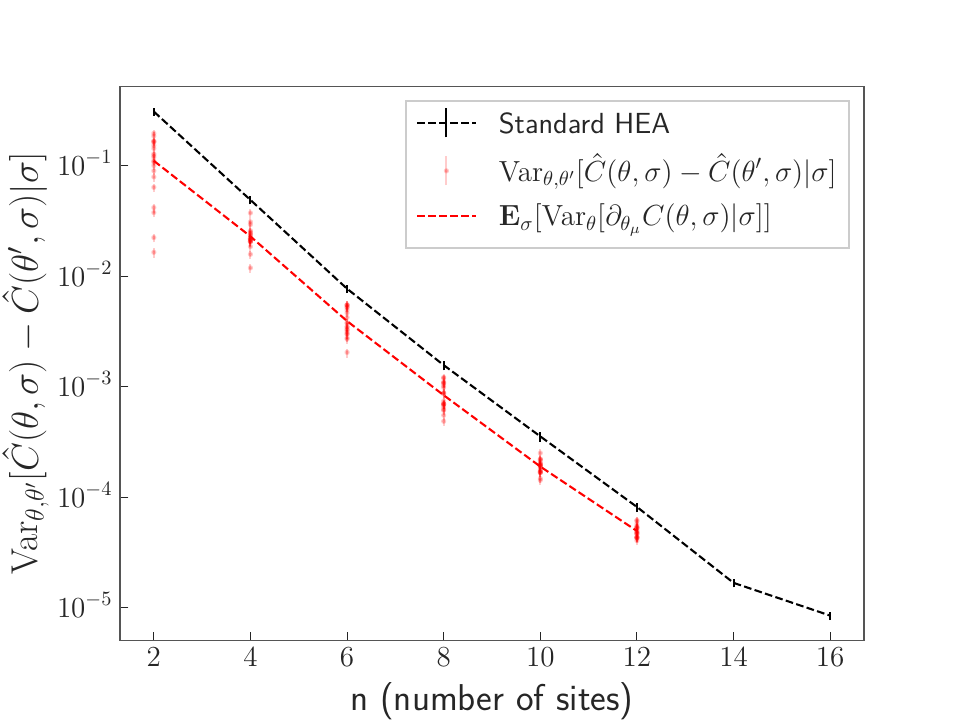}
        \caption{
        }
        \label{fig:parameter_concentration_vqe}
    \end{subfigure}
    \caption{
        VQE cost anti-concentration in the presence of untrainable parameters.
        (\ref{fig:cost_anticoncentration_vqe})
        $\mathrm{Var}_{\theta,\sigma}[C]$ as a function of $n$ for the VQE experiment.
        We can see cost function anti-concentration as the variance is lower
        bounded by $\Theta(1/n)$ (shown in Appendix \ref{sec:methods}).
        (\ref{fig:parameter_concentration_vqe})
        Estimating
        $\mathbb{E}_\sigma[\mathrm{Var}[\partial_{\theta_\mu} C(\theta,\sigma)|\sigma]]$
        by taking $20$ random samples of $\sigma$.
    }
\end{figure*}

\textit{VQE.}
Our first application is ground-state preparation of the 1D Ising Hamiltonian
\begin{equation}
    H = \frac{1}{n}\left( \sum_{i=1}^{n-1}{Z_i Z_{i+1}} + \sum_{i=1}^{n} X_i \right)
    \label{eqn:ising_hamiltonian}
\end{equation}
with VQE using a 1D linearly connected hardware-efficient ansatz (HEA)
\cite{kandala2017hardware} as $U(\theta)$ with $L=20$ layers for $n\leq16$, even $n$.
For the non-unitary part $\mathcal{E}$ we apply a single layer of the feedforward gadget
of \cite{deshpande2024dynamic} (Fig.~\ref{fig:purified_feedforward_gadget}) to each qubit,
at the end of the circuit, with an independent parameter $\sigma_i$ per site gadget.
We refer to the $\sigma=0$ circuit (unitary $U(\theta)$) as the ``Standard HEA"
and the $\sigma\neq0$ version as ``DC-HEA".

For calculating the cost function variance, we sample 200 $(\theta, \sigma)$
parameter instantiations for each data point. Each expectation value is evaluated with
$4096$ shots.
The error bars are $95\%$ confidence intervals evaluated using bootstrapping with
$10,000$ re-samples.
Fig.~\ref{fig:cost_anticoncentration_vqe} shows that the cost variance is lower-bounded by
$\Theta(1/n)$, in agreement with the bound from
Lemma~\ref{lemma:pauli_path_lemma} (derived in Appendix \ref{sec:methods}).
Yet, as Fig.~\ref{fig:parameter_concentration_vqe} shows, the partial gradients of $\theta$
still decay exponentially with $n$.
Here we calculate the statistic $D(\sigma)$ with $20$ uniformly randomly chosen
$\sigma$ values for each $n$.
The cost thus anti-concentrates while the original parameters stay untrainable,
raising concerns about the effectiveness of approaches that insert dynamic gadgets
at the end of the circuit to avoid BP like in \cite{deshpande2024dynamic}.

\textit{QAOA.}
To test the same picture under weaker assumptions: in particular dropping the
locally scrambling assumption on $U(\theta)$ from \cite{deshpande2024dynamic}
and allowing correlated parameters;
we turn to QAOA \cite{farhi2014quantum} and use it to solve the Max-Cut problem.
Max-Cut QAOA takes a graph $G(V,E)$ as input and seeks a partition of $V$ into two sets
that maximizes the number of crossing edges; each partition (cut)
is encoded as an $n=|V|$-bit string.
The fraction of edges cut by a basis state is read out by the diagonal observable
\begin{equation}
    H = \sum_{(u,v)\in E}\frac{I-Z_u Z_v}{2m}
    \label{eqn:qaoa_observable}
\end{equation}
where $m=|E|$.
In its standard formulation, the circuit is an $L$-layer ansatz given by the unitary
\begin{equation}
    \hat{U}(\theta) = 
            e^{-i \theta^{(x)}_L H_M} e^{-i \theta^{(z)}_L  H_P} ...
            e^{-i \theta^{(x)}_1 H_M} e^{-i \theta^{(z)}_1  H_P}
\end{equation}
where $H_P=\sum_{(u,v)\in E}Z_uZ_v$ is the ``problem" Hamiltonian
(which is $H$ from Eq. \ref{eqn:qaoa_observable} up to scaling and a shift)
and $H_M=\sum_{u\in V}X_u$ is the ``mixer" Hamiltonian.
The initial state is the uniform superposition
$\rho_0=|+\rangle \langle+|^{\otimes n}$, which is the ground state of $H_M$.
Measuring the final state $\hat{U}(\theta) \rho_0 \hat{U}(\theta)^{\dagger}$
in the computational basis yields candidate cuts.

\begin{figure*}[t]
    \begin{subfigure}[t]{0.45\textwidth}
        \centering
        \includegraphics[height=1.25in]{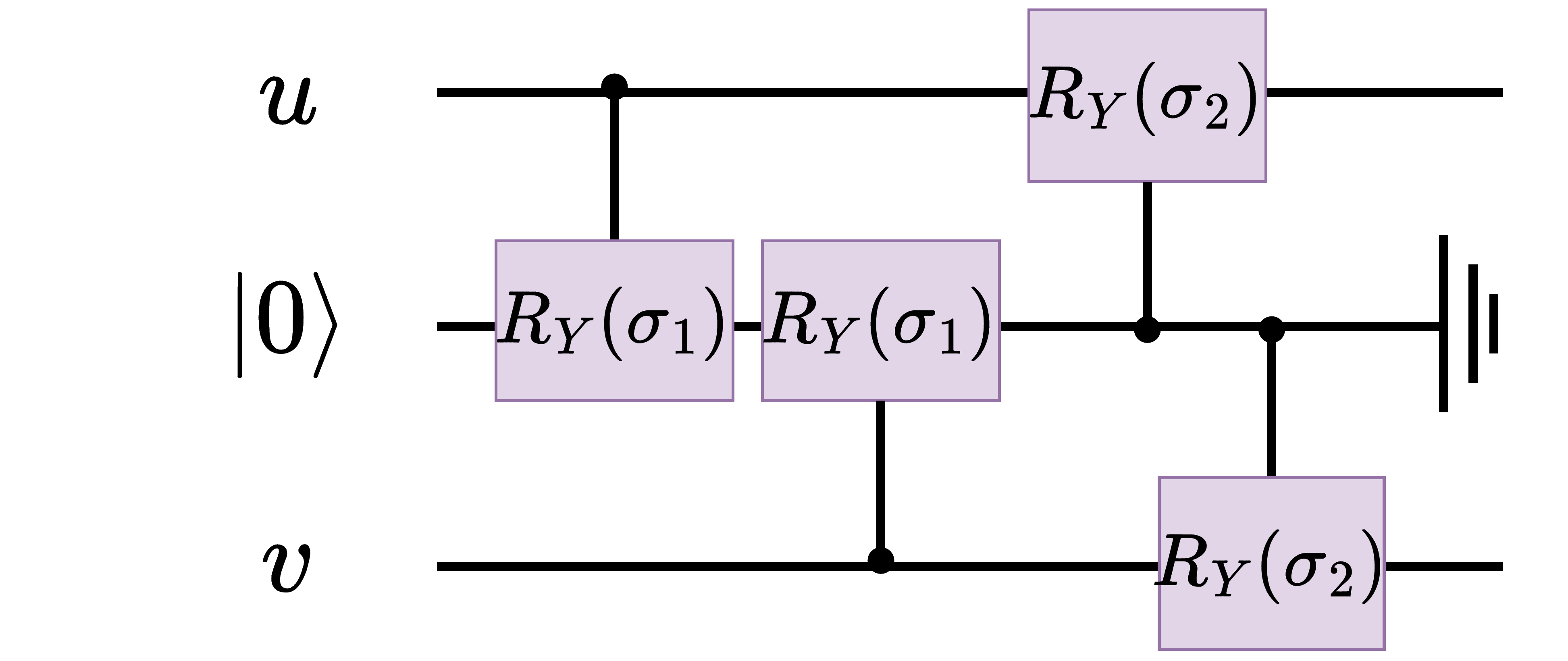}
        \caption{}
        \label{fig:qaoa_gadget_1}
    \end{subfigure}
    \hspace{10pt}
    \begin{subfigure}[t]{0.45\textwidth}
        \centering
        \includegraphics[height=1.25in]{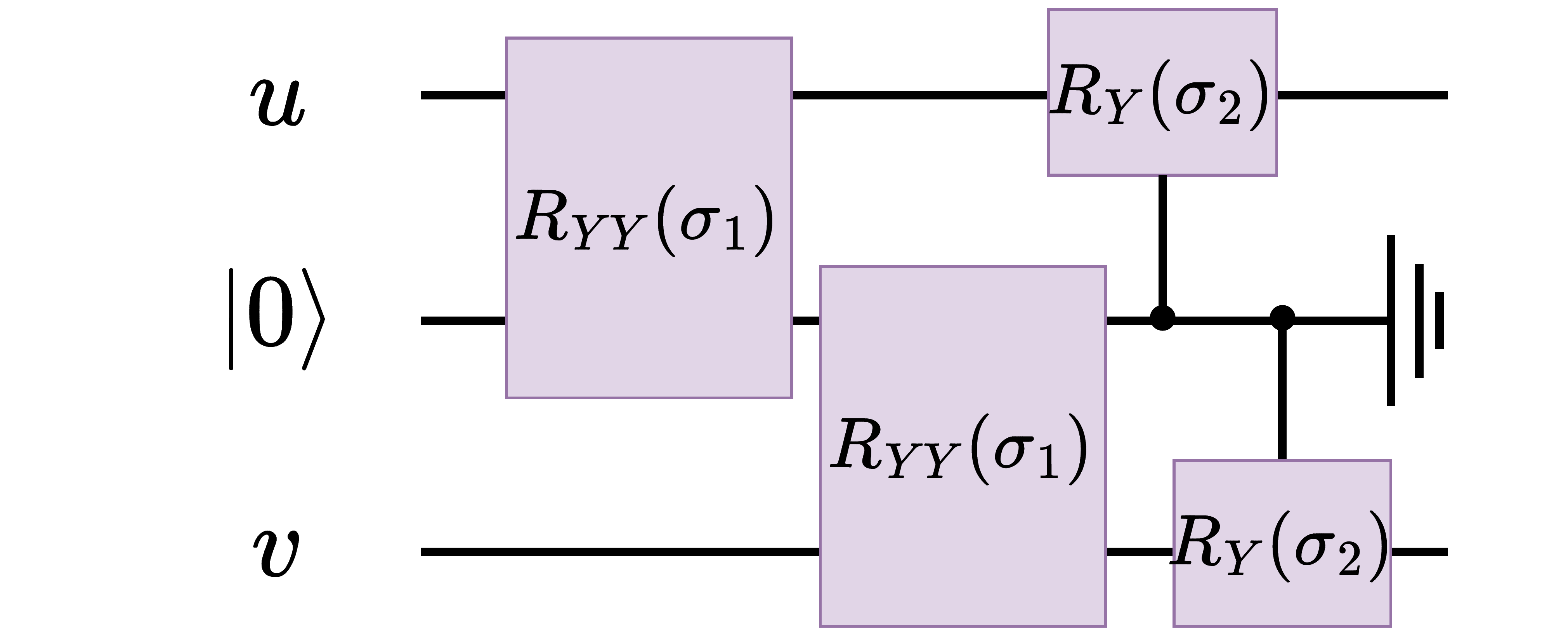}
        \caption{
        }
        \label{fig:qaoa_gadget_2}
    \end{subfigure}
    \caption{
        Max-Cut QAOA gadgets used in our experiments.
        The first pair of gates are ``entanglers", parameterized by $\sigma_1$,
        and the second pair are ``feedforwards", parameterized by $\sigma_2$.
        These gadgets are inserted between every pair of adjacent nodes $(u,v)\in E$
        in the graph $G(V,E)$.
        (\ref{fig:qaoa_gadget_1})
        The first QAOA gadget. This gadget doesn't commute with gadgets on adjacent edges.
        (\ref{fig:qaoa_gadget_2})
        In the second QAOA gadget, the entangler controlled-Y rotations from
        the first gadget are replaced with $YY$ Pauli rotation gates.
    }
    \label{fig:qaoa_gadgets}
\end{figure*}

We introduce the two kinds of edge gadgets shown in Figs.~\ref{fig:qaoa_gadget_1}
and \ref{fig:qaoa_gadget_2}, and observe cost anti-concentration when using the first
kind but not the second.
We provide an analytical explanation of this observation in Appendix \ref{sec:methods}.
A copy of the gadget $\mathcal{E}_e(\sigma_1,\sigma_2)$ is appended to every edge $e\in E$.

Gadgets of the first kind (Fig.~\ref{fig:qaoa_gadget_1}) do not commute with
gadgets on adjacent edges sharing a vertex,
so the order in which they are appended matters and induces a permutation
$\pi\in \mathcal{S}_m$ of the edges.
We remove this ordering bias by symmetrizing the overall channel across all edge
permutations,
\begin{equation}
    \mathcal{E}(\sigma) =
    \frac{1}{m!}\sum_{\pi\in\mathcal{S}_m}
    \mathcal{E}_{\pi(m)}(\sigma_1,\sigma_2)\circ...\circ\mathcal{E}_{\pi(1)}(\sigma_1,\sigma_2)
\end{equation}
which in practice (and in simulation) amounts to uniformly randomly reshuffling the edge
order every few shots.
All gadgets share the same pair of parameters, analogous to the shared parameters in standard QAOA. This choice is motivated by the observed transferability of QAOA parameters across graph instances, which can yield good approximation ratios~\cite{brandao2018fixed,shaydulin2023parameter}. Accordingly, the gadget parameters $(\sigma_1,\sigma_2)$ are shared across all edges to enable analogous parameter transferability.

For our experiments we sample Erd\H{o}s-R\'enyi graphs $G(n,\frac{1}{2})$
(each possible edge is included independently with probability one half)
and set $L=10$ for $n\leq10$,
(one graph for each even $n$).
We sample $300$ $(\theta,\sigma)$ parameter instantiations per data point,
and evaluate the expectation value for each using $m^2$ edge permutation samples,
with four shots for every permutation.
Fig.~\ref{fig:cost_anticoncentration_qaoa} shows clear cost anti-concentration;
quadrupling the number of shots per permutation,
and doubling the number of permutation samples independently,
leaves the plot essentially unchanged.
This confirms that the variance lower bound is not a sampling artifact.
To probe the effect of deepening $\hat{V}(\sigma)$
(from Eq.~\ref{eqn:purified_dpqc}), we append $f$ standard QAOA layers
$e^{-i \theta^{(x)}_j H_M} e^{-i \theta^{(z)}_j  H_P}$ after the
non-unitary channel $\mathcal{E}(\sigma)$ and plot the variance lower bound
$\min_{n} \mathrm{Var}_{\theta,\sigma}[C(\theta,\sigma)]$ as a function of $f$.
Fig.~\ref{fig:variance_lower_bound_qaoa} shows this bound decaying roughly as
$\Theta(d^{-\lambda f})$ for some constants $d>1,\lambda>0$ (evenly spaced lines on a
semi-log plot).
These results are consistent with the Pauli path observation of
Sec.~\ref{sec:pauli_path_analysis} and with the ``feedforward distance" analysis of
\cite{deshpande2024dynamic}, and hint at a general trend beyond the aforementioned
assumptions: the dynamic layer cannot be inserted far from the final observable
measurements without losing its effect; in line with Lemma ~\ref{lemma:sparse_dpqc}.

As in the VQE case, the statistic $D(\sigma)$ concentrates,
exhibiting behavior nearly identical to that shown in
Fig.~\ref{fig:parameter_concentration_vqe};
the corresponding figure is therefore omitted.
To conclude untrainability from the concentration of this statistic,
we invoke the proofs of~\cite{arrasmith2022equivalence},
which rely on the parameter-shift rules~\cite{mitarai2018quantum,wierichs2022general}.
These rules apply only when each gate carries its own free parameter,
whereas in QAOA a single parameter drives $O(\mathrm{poly(n)})$ gates.
We therefore consider a free-parameterized version of the cost,
in which every gate is assigned an independent angle; the true QAOA cost is
recovered by setting all angles driven by a common parameter equal to one another.
By the multi-variable chain rule, each true QAOA gradient is then a sum of
$O(\mathrm{poly}(n))$ free-parameter partial derivatives.
Since each such term exhibits exponentially suppressed sample variance,
so does their sum, and the concentration carries over to the true QAOA cost.

\begin{figure*}[htpb]
    \begin{subfigure}[t]{0.45\textwidth}
        \centering
        \includegraphics[height=2.25in]{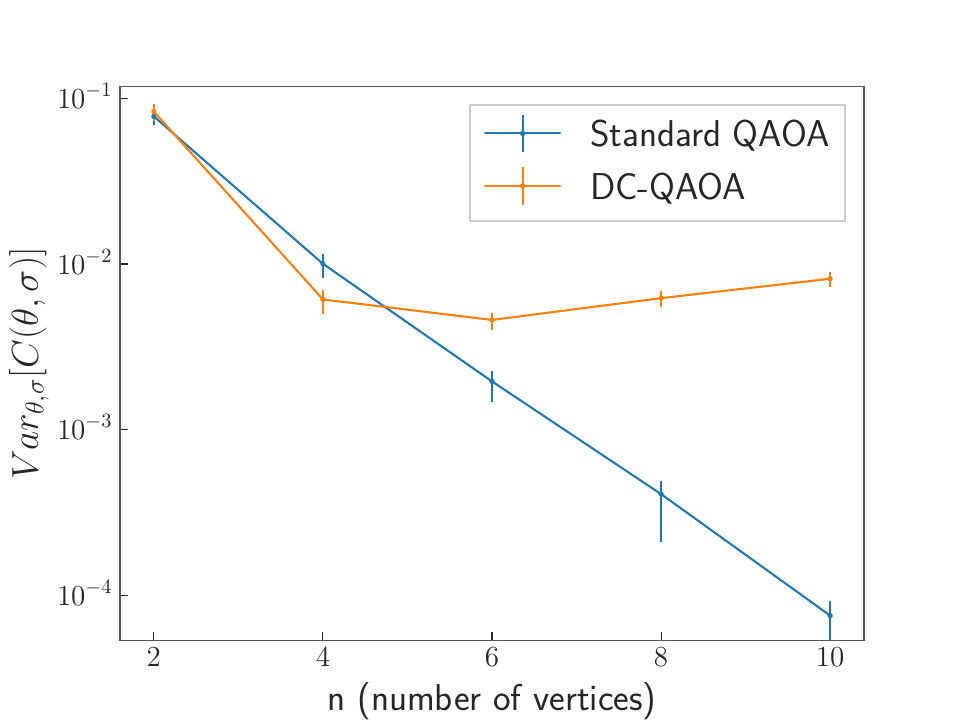}
        \caption{}
        \label{fig:cost_anticoncentration_qaoa}
    \end{subfigure}
    \hspace{10pt}
    \begin{subfigure}[t]{0.45\textwidth}
        \centering
        \includegraphics[height=2.25in]{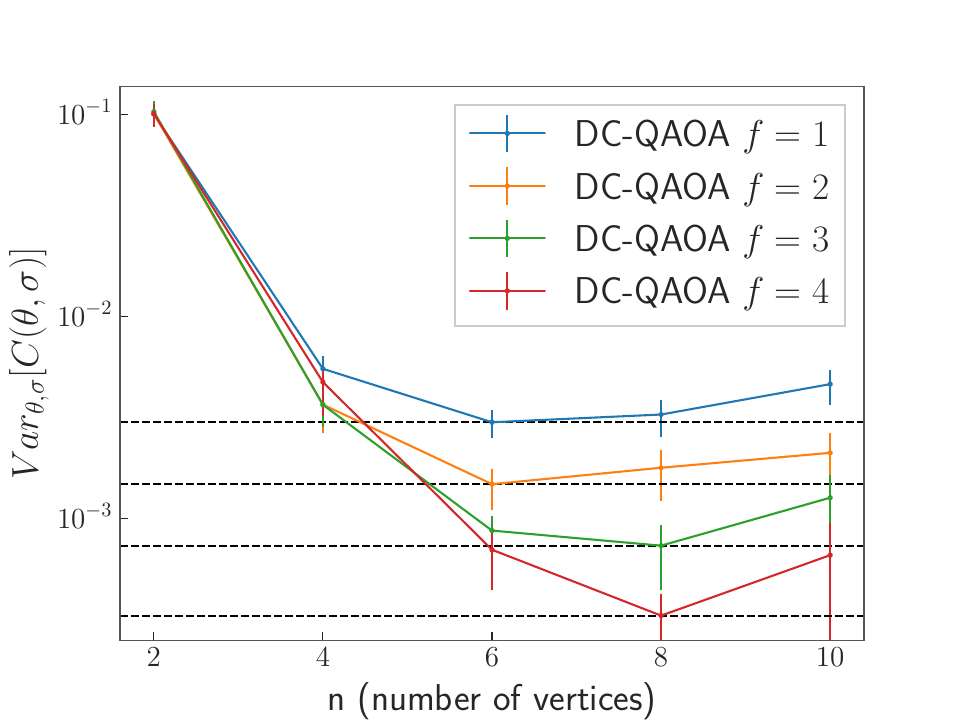}
        \caption{}
        \label{fig:variance_lower_bound_qaoa}
    \end{subfigure}
    \caption{
        (\ref{fig:cost_anticoncentration_qaoa})
        $\mathrm{Var}_{\theta,\sigma}[C]$ as a function of $n$
        for the QAOA experiment with gadgets shown in Fig.~\ref{fig:qaoa_gadget_1}.
        (\ref{fig:variance_lower_bound_qaoa})
        Variance lower bound after varying ``feedforward distance" $f$,
        which is the maximum distance between an observable and the nearest
        dynamic gadget in its light cone. In our QAOA experiment, it is the number of layers
        between the final measurement and the dynamic gadget layer.
    }
\end{figure*}

The second kind of gadget (Fig.~\ref{fig:qaoa_gadget_2}) is designed so
that it commutes with gadgets on adjacent edges. But we observe cost concentration
in the DPQC constructed from that kind of gadget layer.

\begin{figure*}[htpb]
    \begin{subfigure}[t]{0.45\textwidth}
        \centering
        \includegraphics[height=2.25in]{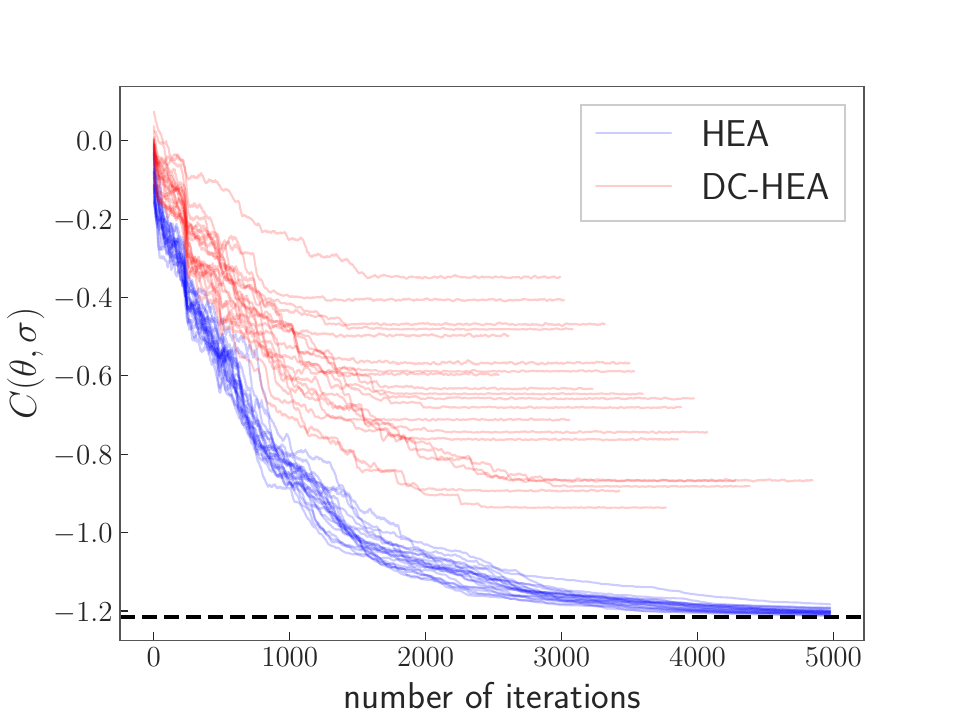}
        \caption{}
        \label{fig:hea_loss_plot}
    \end{subfigure}
    \hspace{10pt}
    \begin{subfigure}[t]{0.45\textwidth}
        \centering
        \includegraphics[height=2.25in]{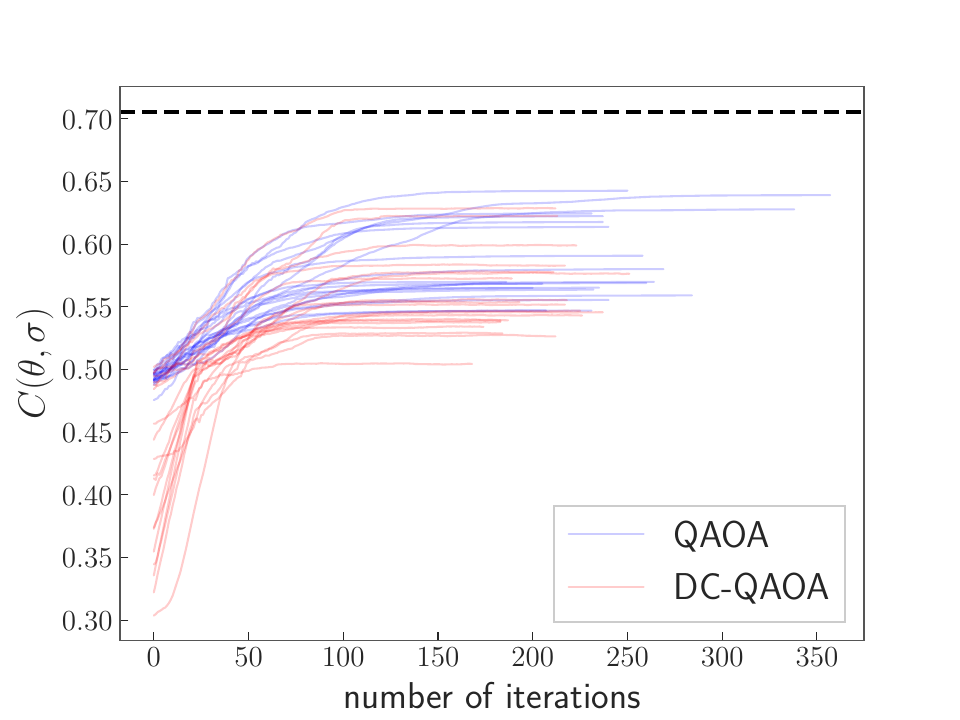}
        \caption{}
        \label{fig:qaoa_loss_plot}
    \end{subfigure}
    \caption{
        Degradation in cost function landscape due to added non-unitary channels
        $\mathcal{E}(\sigma)$.
        (\ref{fig:hea_loss_plot}) is the plot for VQE, and (\ref{fig:qaoa_loss_plot})
        is for QAOA.
        The red lines are dynamic circuit variants and the blue lines
        are the standard unitary variants.
        These plots show the optimization trajectory of 20 uniformly randomly
        initialized parameters.
        We take a moving average over 25 iterations to remove high frequency
        oscillations.
        The black dashed line represents the optimal cost.
    }
    \label{fig:loss_plots}
\end{figure*}

Finally, we train the standard and dynamic variants, ``HEA" against ``DC-HEA" on a
six-qubit VQE instance, and ``QAOA" against ``DC-QAOA" on an eight-qubit instance.
We use the gradient-free optimizer COBYLA \cite{zaikun_zhang_2023_8052655}.
As Fig.~\ref{fig:loss_plots} shows, the dynamic layer brings no improvement to the
cost landscape as seen from randomly sampled optimization trajectories.
To further evaluate the contribution of the dynamic layer, we re-evaluate the
cost after removing it post training, i.e., by setting $\sigma=0$ at inference time.
For VQE, removing $\mathcal{E}$ reduces the energy gap relative to the true ground-state
energy of the 1D Ising model by roughly $23\%\pm9\%$ over 20 random parameter optimization
results, suggesting that the dynamic layer was not helpful in finding the ground state.
For QAOA, removing the layer changes the cut fraction by less than $3\%$ on average
across 20 random initializations, suggesting that the trained dynamic layer remains
close to the identity.
Moreover, the $\theta$ gradients become untrainable for large $n$, showing behavior
similar to the VQE results in Fig.~\ref{fig:parameter_concentration_vqe}.
Therefore, ``DC-QAOA'' does not provide an improvement over standard QAOA. 

As an additional experiment, we collect the best-performing DC-QAOA parameters 
$(\theta,\sigma)$, trained on a random eight-vertex graph,
and reuse them across Erd\H{o}s--R\'enyi graphs of sizes $6$, $8$, and $10$,
with $20$ graphs of each size.
We observe approximation ratios above $80\%$ on average in each case,
suggesting that the dynamic parameters are transferable across graphs.
However, this transferability does not address the untrainability of the
$\theta$ parameters and therefore does not change the conclusion that
DC-QAOA provides no trainability advantage over standard QAOA.

\section{Discussion}
\label{sec:discussion}

We introduce DPQC as a unifying formalism that encompasses the disparate non-unitary
strategies proposed for BP mitigation, and analyze their trainability.
Our central finding is a tension between expressivity and trainability that in general cannot be resolved solely through the insertion of dynamic gadgets. Specifically, faithful DPQCs inherit the BPs of the unitary PQCs from which they are constructed, as established in Lemma~\ref{lemma:faithful_dpqc}. Moreover, even \emph{unfaithful} DPQCs that insert only $O(1)$ gadget layers leave $\Omega(L)$ parameters untrainable whenever the underlying ansatz is translationally invariant and any sub-array forms a $2$-design, as shown in Lemma~\ref{lemma:sparse_dpqc}.
Together these results rule out a broad class of existing constructions as viable BP
mitigation strategies.

The Pauli path analysis underlying Lemma~\ref{lemma:pauli_path_lemma}, carried out on the
purified channel, characterizes the effect of inserting non-unitary layers not too far
from the end of the circuit and clarifies why such constructions may fail to restore
trainability of the parameters coming before it.
Specifically, a dynamic gadget anti-concentrates the cost function values by
generating high-variance Pauli strings supported entirely on the ancilla qubits,
whose coefficients depend only on $\sigma$ and are not scrambled by $U(\theta)$.
This produces the \emph{illusion} of a trainable landscape: the cost function retains substantial variation, while the $\theta$ gradients remain exponentially suppressed.
Our numerical experiments on both 1D Ising model VQE
and graph Max-Cut QAOA confirm this mechanism:
cost anti-concentration coexists with untrainable $\theta$ parameters.
Moreover, the part of the circuit with trainable parameters is subject to
classical analytic methods like Pauli paths.
This behavior is conceptually related to recent results on noise-induced shallow circuits
\cite{mele2026noise,cerezo2025does}, where the same mechanism that suppresses BPs renders
the circuit classically simulable.
Finally, joint optimization of $(\theta, \sigma)$ yields no improvement over
standard unitary ansatzes even for small system sizes, at least not with our
gadgets/ansatzes.

\emph{Outlook}:
Our results suggest that BP mitigation via dynamic circuits is at least as hard as
designing BP-free unitary ansatzes from first principles.
At a high level, this follows from Stinespring's dilation theorem
~\cite{stinespring_dilation}, since a purified DPQC is itself a unitary circuit.
Our analysis further supports this conclusion.
To achieve genuine trainability of $\theta$ in a DPQC, one must sacrifice faithfulness and
insert gadget layers frequently enough that the augmented landscape avoids $2$-design
behavior from contiguous sub-circuits.
A promising direction is to design DPQCs whose gadget/circuit structure
is adapted to the target Hamiltonian rather than appended generically or randomly, using the
Pauli path framework developed here as a design tool.
Whether such construction can be simultaneously expressive, trainable in $\theta$, and
not classically simulable in its trainable directions remains the central open question for
the practical utility of dynamic circuits in VQAs.

\section*{Data availability}
\label{sec:data_availability}
Simulation data and the code used to generate it is freely available on GitHub
\href{https://github.com/sumeetshirgure/dynbp}{https://github.com/sumeetshirgure/dynbp}.
The symbolic ``Pauli-Heisenberg evolution" engine is also provided
as a python package ``sympauli".
\href{https://pypi.org/project/sympauli/}{https://pypi.org/project/sympauli/}

\begin{acknowledgments}
The authors would like to thank Marco Cerezo for a stimulating discussion on classical simulability.
S.S. is supported by the University of Central Florida ORCGS Doctoral Fellowship award.
This research used resources of the National Energy Research Scientific Computing Center (NERSC),
a Department of Energy User Facility using NERSC award DDR-ERCAP 0038372.
\end{acknowledgments}

\section*{Author contributions}
S.N. proposed the initial idea of investigating BP mitigation in QAOA using dynamic circuits.
S.S. proposed all of the concrete ideas, proofs, and theoretical analysis.
E.K. and S.N. provided feedback and guidance. 
The code for sympauli was written by S.S. with the help of
generative AI agent ``Claude" from Anthropic [Sonnet 4.6 accessed May 2026],
and all relevant code was reviewed and validated by S.S.
The manuscript was written and reviewed by all authors.

\section*{Competing interests}
The authors declare no competing interest.

\bibliography{ref}

\appendix
\begin{widetext}

\section{Methods}
\label{sec:methods}
In this section we present the proofs and theoretical analyses presented in the main article.
We start with the proof of lemma \ref{lemma:faithful_dpqc}.

\begin{proof}
(Lemma \ref{lemma:faithful_dpqc})
Recall that $\mathcal{U}(\theta,\sigma)$ is faithful to $U(\theta)$
if for the respective $(H,\rho_0)$, the cost function deviates uniformly pointwise
by an exponentially small amount
\begin{equation}
    |C(\theta,\sigma)-C(\theta,0)|\in O(c^{-n}) \ \ \ \ \forall \theta_\mu\in[-\pi,\pi),\
    \forall \sigma: ||\sigma|| < \epsilon \in \Theta \left(\frac{1}{poly(n)} \right)
\end{equation}
for some $c>1$.
Note that we can bound the differences in the partial gradients using the triangle inequality
\begin{align}
    |\partial_{\theta_\mu}C(\theta,\sigma)-\partial_{\theta_\mu}C(\theta,0)| \ 
    =\ \frac{1}{2}|(C(\theta+\pi\mathbf{e}_\mu/2,\sigma)-C(\theta+\pi\mathbf{e}_\mu/2,0))
     -(C(\theta-\pi\mathbf{e}_\mu/2,\sigma)-C(\theta-\pi\mathbf{e}_\mu/2,0))| \\
    \leq \frac{1}{2}|C(\theta+\pi\mathbf{e}_\mu/2,\sigma)-C(\theta+\pi\mathbf{e}_\mu/2,0)|
   +\frac{1}{2}|C(\theta-\pi\mathbf{e}_\mu/2,\sigma)-C(\theta-\pi\mathbf{e}_\mu/2,0)| \in O(c^{-n})
\end{align}
$\mathbf{e}_\mu$ is the basis vector in the $\mu$ direction,
meaning $\theta_\mu$ is shifted by an amount of $\pi/2$ in either direction.
The parameter shift finite difference formula comes from
\cite{mitarai2018quantum,wierichs2022general}.
To see that parameter shift rules apply in the presence of dynamic circuit operations,
we can think of purifying each CPTP map.
Alternatively, we can use the lemma provided in the supplemental material of
\cite{yan2025variational}.
This bound on the difference of partial gradients implies that adding the $\sigma$
parameters doesn't change the asymptotic behaviour of the partial gradients along $\theta_\mu$ in the small $\sigma$ regime.
\end{proof}

Lemma \ref{lemma:faithful_dpqc} only states untrainability of $\theta_\mu$,
but to first order, we can also similarly bound the $\sigma$ gradients.
We can do this by writing the Taylor series of
$C(\theta, \sigma+\delta)$ around $\sigma$ as
\begin{equation}
    C(\theta, \sigma+\delta) - C(\theta, \sigma) =
    \sum_\nu{\delta_\nu\partial_{\sigma_\nu}C} + O(||\delta||^2)
    \label{eqn:taylor}
\end{equation}
If any of the partial gradients $\partial_{\sigma_\nu}C$
becomes $\Omega \left(\frac{1}{poly(n)}\right)$,
then perturbing $\sigma_\nu$ by an amount $\delta_\nu$ that is itself
not exponentially small might lead to a contradiction because we can upper
bound the L.H.S from the faithfulness constraint,
and lower bound the R.H.S of Eq. \ref{eqn:taylor} by the above assumption:
\begin{align}
    |L.H.S| = |(C(\theta,\sigma+\delta)-C(\theta,0))-(C(\theta,\sigma)-C(\theta,0))| \in O(c^{-n}) \\
    |R.H.S| \approx |\sum_\nu\delta_\nu \partial_{\sigma_\nu}C(\theta,\sigma)|
        =\sum_\nu|\delta_\nu| |\partial_{\sigma_\nu}C(\theta,\sigma)|
            \in \Omega\left(\frac{1}{poly(n)}\right)
    \label{eqn:delta_choice}
\end{align}
Equation \ref{eqn:delta_choice} follows from the ability to choose the signs
of $\delta_\nu$. We cannot choose exponentially small perturbations
$\delta$ due to the resulting explosion in sampling complexity.
However, there could be cancellations if we include the second or higher order
partial derivatives along $\sigma$.
The analysis of these higher order derivatives requires further assumptions
about the structure of the dynamic gadgets, which we aim to avoid
to keep the statement generic, and hence Lemma \ref{lemma:faithful_dpqc} only
deals with $\theta_\mu$ partial gradients. But this first order analysis
only strengthens the notion of using the $\sigma$
directions to encode the solution in the extended landscape.

Next we look at the variance decomposition formula used in the following
proofs of lemmas \ref{lemma:sparse_dpqc} and \ref{lemma:pauli_path_lemma}.
\begin{proof}
(Variance decomposition in equation \ref{eqn:law_of_total_variance})
From the law of total variance we have :
\begin{equation}
    \mathrm{Var}_{\theta,\sigma}[\partial_{\theta_\mu}C]
        = \mathbb{E}_\sigma[\mathrm{Var}_\theta[\partial_{\theta_\mu}C|\sigma]] +
            \mathrm{Var}_\sigma[\mathbb{E}_\theta[\partial_{\theta_\mu}C|\sigma]]
\end{equation}
Since $C$ is periodic in $\theta$, the last term (variance of conditional expectations)
vanishes because the anti-derivative of the partial gradient of $C$
is evaluated at identical endpoints.
\end{proof}

For an introduction to t-designs and Haar measure tools, which are required by the
proofs of lemmas \ref{lemma:sparse_dpqc} and \ref{lemma:pauli_path_lemma},
we refer the reader to \cite{mele2024introduction}.

\begin{proof}
(Lemma \ref{lemma:sparse_dpqc})
Factorize $\mathcal{U}(\theta,\sigma)$ as
\begin{equation}
    \mathcal{U}(\theta,\sigma)=\mathcal{E}_R(\phi_R) \circ
    U_{v_j}(\theta_{v_j})...U_{v_{j-1}+1}(\theta_{v_{j-1}+1})
    \circ \mathcal{E}_L(\phi_L)
\end{equation}
by absorbing the channels before (after) the
$\hat{U}_{w_j}(\theta_{w_j}) = \hat{U}_{v_j}(\theta_{v_j})...\hat{U}_{v_{j-1}+1}(\theta_{v_{j-1}+1})$
block into a single map $\mathcal{E}_L$ ($\mathcal{E}_R$),
and collecting the respective parameters into $\phi_L$ ($\phi_R$).

Let $H_R=\mathcal{E}^\dagger_R(\phi_R)[H]$ and
let $\rho_L=\mathcal{E}_L(\phi_L)[\rho_0]$ for some parameter instantiations.
Consider any parameter $\theta_\mu$ in $\hat{U}_{w_j}$
parameterizing the gate
$\hat{U}_\mu(\theta_\mu)=e^{-i \frac{\theta_\mu}{2}P_\mu}\hat{W}_\mu$ with $P_\mu^2=I$.
Factorize the blocks before and after this gate as $\hat{U}_{w_j}=\hat{U}_+\hat{U}_\mu \hat{U}_-$.
By the sub-array 2-design assumption, at least one of $\hat{U}_+$ or $\hat{U}_-$
has length $\Omega(L)$ and is hence a 2-design.
Consider the variance decomposition
$\mathrm{Var}_{\theta_{w_j},\phi_L,\phi_R}[\partial_{\theta_\mu}C]
=\mathbb{E}_{\phi_L,\phi_R}[\mathrm{Var}_{\theta_{w_j}}
[\partial_{\theta_\mu}C|\phi_L,\phi_R]]$.
We adapt the calculation from Equation 7 in \cite{mcclean2018barren}
and plug in $H_R$ and $\rho_L$ into the formulas to bound
this conditional variance for the three cases :

\begin{equation}
    \mathrm{Var}_{\theta_{w_j}}[\partial_{\theta_\mu}C|\phi_L,\phi_R] \approx
    \begin{cases}
        \frac{\mathrm{Tr}[\rho_L^2]}{4^n-1}
            \mathbb{E}_{u\sim \hat{U}_+}\left[||[P_\mu,u^\dagger H_R u]||_F^2\right]
        & \text{if $\hat{U}_-$ forms a 2-design} \\ \\
        \frac{\mathrm{Tr}[H_R^2]}{4^n-1}
            \mathbb{E}_{u\sim \hat{U}_-}[||[P_\mu,u \rho_L u^\dagger]||_F^2]
        & \text{if $\hat{U}_+$ forms a 2-design} \\ \\
        \frac{\mathrm{Tr}[H_R^2]\mathrm{Tr}[\rho_L^2]\mathrm{Tr}[P_\mu^2]}{2^{3n-1}}
        & \text{both $\hat{U}_+$ and $\hat{U}_-$ form 2-designs}
    \end{cases}
\end{equation}

where $||A||_F^2=\mathrm{Tr}[A^\dagger A]$ denotes the Frobenius\
norm derived from the Hilbert-Schmidt inner product.
Note that $\mathrm{Tr}[\rho_L^2] \leq 1$ for all cases.
We bound the Frobenius norm of the commutator
$||[A,B]||_F^2 \leq 4 ||A||_F^2 ||B||_{op}^2$, where $||.||_{op}$ denotes
operator norm \cite{bhatia97}.
For the first case, we can bound $||[P_\mu,u^\dagger H_R u]||_F^2$ by
$4 ||P_\mu||_F^2 ||u^\dagger H_R u||_{op}^2$.
$||P_\mu||_F^2=\mathrm{Tr}[P_\mu^2]=\mathrm{Tr}[I]=2^{n}$.
$||u^\dagger H_R u||_{op}^2=||H_R||_{op}^2 \leq ||H||_{op}^2$,
where the last inequality follows from the fact that the adjoint
map $\mathcal{E}_R^\dagger$ is unital for any CPTP map $\mathcal{E}_R$,
and unital maps are contractive in the operator norm \cite{bhatia97}.
(This property of unital maps when acting on a Hermitian operator $A$ can be
proven by writing the PSD matrix inequalities $-||A||_{op} I\leq A \leq ||A||_{op} I$
and applying the unital map on all three sides and finally using the defining fact
of unital maps: mapping the identity operators on both ends back to identity.)
So the conditional variance is upper bounded by
$\frac{4\cdot 2^{n} ||H||_{op}^2}{4^n-1} \in O\left(\frac{||H||_{op}^2}{2^n}\right)$.
For the second case, we first apply the commutator bound such that
$||[P_\mu,u \rho_L u^\dagger]||_F^2
    \leq 4 ||P_\mu||_{op}^2||u \rho_L u^\dagger||_F^2 = 4||\rho_L||_F^2 \leq 4$.
Then we bound $||H_R||_F^2 \leq 2^n||H_R||_{op}^2 \leq 2^n ||H||_{op}^2$
giving us an upper bound of $O\left(\frac{||H||_{op}^2}{2^n}\right)$ for the conditional variance.
For the third case we again get $O\left(\frac{||H||_{op}^2}{2^n}\right)$ using the substitutions above.
\end{proof}

\begin{proof}
    (Lemma \ref{lemma:pauli_path_lemma})
    Write the cost function as
    \begin{equation}
        C(\theta, \sigma)
        = \mathrm{Tr}\left[
        \left( \sum_P{f_P(\sigma)P} \right) \cdot
        (\rho(\theta)_s\otimes|0\rangle\langle0|_a)
        \right]
    \end{equation}
    where $\rho(\theta)=\hat{U}(\theta) \rho_0 \hat{U}(\theta)^\dagger$.
    Separating the ancilla supported terms $I_s\otimes P_a$ from the rest we get :
    $C(\theta, \sigma) = F(\sigma) + R(\theta, \sigma)$
    where $F(\sigma) = \sum_{P_a} f_{P_a}(\sigma) \mathrm{Tr}[P_a |0\rangle\langle0|_a]$
    and
    \begin{equation}
        R(\theta,\sigma) =
        \sum_{Q=Q_s\otimes Q_a,Q_s\neq I_s}{
            f_Q(\sigma)
            \mathrm{Tr}[Q_s \rho(\theta)]
            \mathrm{Tr}[Q_a |0\rangle\langle0|_a]
        }
        \label{eqn:pauli_path_residual}
    \end{equation}
    Next, we observe that by the 1-design property of $U(\theta)$,
    \begin{equation}
        \mathbb{E}_\theta[\mathrm{Tr}[Q_s \hat{U}(\theta) \rho_0 \hat{U}^\dagger(\theta)]]
        = \frac{\mathrm{Tr}[Q_s]\mathrm{Tr}[\rho_0]}{2^n} = 0
    \end{equation}
    because $Q_s\neq I_s$ is traceless.
    $\implies \mathbb{E}_\theta[R(\theta,\sigma)|\sigma]=0$.
    Finally, we can lower bound the variance of $C$ as
    \begin{equation}
        \begin{split}
            \mathrm{Var}_{\theta,\sigma}[C(\theta,\sigma)]
            &= \mathrm{Var}_\sigma[\mathbb{E}_\theta[C(\theta, \sigma)|\sigma]]
             +\mathbb{E}_\sigma[\mathrm{Var}_\theta[C(\theta, \sigma)|\sigma]] \\
            &\geq \mathrm{Var}_\sigma[\mathbb{E}_\theta[C(\theta, \sigma)|\sigma]] \\
            &= \mathrm{Var}_\sigma[F(\sigma)+\mathbb{E}_\theta[R(\theta,\sigma)|\sigma]] \\
            &= \mathrm{Var}_\sigma[F(\sigma)]
        \end{split}
    \end{equation}
\end{proof}

Thus, the cost is anti-concentrated if the sum of ancilla supported Pauli
term coefficients is anti-concentrated.
If we further assume that $U(\theta)$ forms a 2-design as well, then
we can improve our lower bound in terms of
$\mathbb{E}_\sigma[F(\sigma)+\mathrm{Var}_\theta[R(\theta,\sigma)]]$
and $f_Q(\sigma)$ by using the 2-design variance formula for the residual
term $R(\theta, \sigma)$ in equation \ref{eqn:pauli_path_residual}.

Next we provide a theoretical lower bound on cost function variance
using Pauli path analysis and with the help of lemma \ref{lemma:pauli_path_lemma}
on our results from section \ref{sec:numerical_analysis}.

Consider the feedforward gadget from \cite{deshpande2024dynamic}
(Fig.~\ref{fig:purified_feedforward_gadget}).
Let $U=\frac{I-iX}{\sqrt{2}}X$, and $R=\exp(-i \sigma X / 2)$.
Order the three qubits as in the figure,
where the first and the last qubits are ancillas and the middle one is the data qubit.
The first gate is $R\otimes I \otimes I$, followed by $\mathrm{CCU}$
(a multi-controlled unitary $U$) followed by $\mathrm{CSWAP}$.
We are interested in the backward evolutions of the Ising Hamiltonian terms
of the form $I\otimes P \otimes I$, where $P\in\{X,Z\}$.
Note that the gadgets across sites commute because they act on disjoint sets
of qubits so we can take the product of the evolved Pauli terms
for the two-body interactions of the form $Z_i Z_{i+1}$ in Eq. \ref{eqn:ising_hamiltonian}.
Denote $\Pi_j=|j\rangle\langle j|$.
First let's conjugate $H_0=I\otimes P\otimes I$ by
$\mathrm{CSWAP} = \Pi_0\otimes I\otimes I+\Pi_1\otimes \mathrm{SWAP}$
to get $H_1$
\begin{equation}
    H_1 = \mathrm{CSWAP}^\dagger H_0 \mathrm{CSWAP}
        = \Pi_0 \otimes P \otimes I + \Pi_1 \otimes I \otimes P
\end{equation}
Next we conjugate $H_1$ with
$\mathrm{CCU} = (I-\Pi_1\otimes\Pi_1)\otimes I + \Pi_1\otimes\Pi_1\otimes U$
to get $H_2$ :
\begin{equation}
    \begin{split}
        H_2 &= \mathrm{CCU}^\dagger H_1 \mathrm{CCU} \\
            &= \Pi_0 \otimes P \otimes I
            + \mathrm{CCU}^\dagger (\Pi_1 \otimes I \otimes P) \mathrm{CCU} \\
            &= \Pi_0 \otimes P \otimes I
            + \Pi_1\otimes\Pi_0\otimes P
            + \Pi_1 \otimes \Pi_1 \otimes (U^\dagger P U)
    \end{split}
\end{equation}
Finally, we conjugate $H_2$ with $R\otimes I \otimes I$ to get $H_3$
\begin{equation}
    \begin{split}
        H_3 &= (R\otimes I \otimes I )^\dagger H_2 (R\otimes I \otimes I) \\
            &= (R^\dagger \Pi_0 R) \otimes P \otimes I
            +  (R^\dagger \Pi_1 R) \otimes [\Pi_0 \otimes P + \Pi_1\otimes(U^\dagger P U)]
    \end{split}
\end{equation}
Writing the projectors as $\Pi_j=\frac{I+(-1)^jZ}{2}$, and collecting Pauli terms
with identity on the middle qubit we have
\begin{equation}
    H_a = \frac{1}{2} (R^\dagger \Pi_1 R) \otimes I \otimes [P +(U^\dagger P U)]
\end{equation}
The coefficients of $H_a$ after projecting to the $|0\rangle$
subspace of the ancillas are given by
\begin{equation}
    f(\sigma) = \frac{1}{2} \mathrm{Tr}[\Pi_0(R^\dagger \Pi_1 R)\Pi_0]
                \mathrm{Tr}[\Pi_0 (P+U^\dagger P U) \Pi_0]
\end{equation}
For $P=X$, $f(\sigma)=0$, and for $P=Z$,
$f(\sigma) = \frac{1}{2}\sin^2(\frac{\sigma}{2})$.
This gives us $F(\sigma)$ as defined in lemma \ref{lemma:pauli_path_lemma}
for the 1D Ising Hamiltonian as follows:
\begin{equation}
    F(\sigma) = \frac{1}{4n}
    \sum_{i=1}^{n-1}
    {\sin^2 \left( \frac{\sigma_i}{2} \right) \sin^2 \left(\frac{\sigma_{i+1}}{2} \right)}
\end{equation}
We can calculate the variance of $F(\sigma)$ by defining i.i.d random variables
$x_i = \sin^2(\sigma_i/2)$ with $\mathbb{E}[x_i] = 1/2, \mathbb{E}[x_i^2] = 3/8$
and so on when $\sigma_i\in[-\pi,\pi)$. Doing the calculation we get
\begin{equation}
    \mathrm{Var}_\sigma[F(\sigma)] = \frac{9n-13}{1024n^2} \sim \Theta(1/n)
\end{equation}
which is shown in Fig.~\ref{fig:cost_anticoncentration_vqe}.
As we can see, this lower bound is tight for $n \geq 10$.

Lemma \ref{lemma:pauli_path_lemma} allows us to reason about shallow circuits
coming after the dynamic layer.
But it can also be helpful in reasoning about deep circuits in some cases.
We give a partial analysis of why the
two kinds of Max-Cut QAOA gadgets in Fig.~\ref{fig:qaoa_gadgets} show different
behavior when it comes to cost anti-concentration.

In the first gadget in Fig.~\ref{fig:qaoa_gadget_1}, we can consider
the last edge $(u,v)$ in the random shuffle.
There is a corresponding $-Z_u Z_v/2m$ term in the Hamiltonian.
Evolving it backwards through the corresponding edge gadget gives us
two terms $I_u I_v X_e$ and $I_u I_v Z_e$ of which the second term
is not annihilated by the $\Pi_0$ ancilla projector. The corresponding
coefficient is as follows:
\begin{equation}
    f(\sigma) = 
    \frac{1}{2m}
    \left[\frac{1}{2}- \sin^2\left(\frac{\sigma_1}{2}\right)\right]
    \sin^2\left(\frac{\sigma_1}{2}\right)
    \sin^2\left(\sigma_2\right)
\end{equation}
This was evaluated using ``sympauli", our software tool.
Even though there are other terms in $F(\sigma)$, unless they
catastrophically cancel this term, we can expect
$\mathrm{Var}_\sigma[F(\sigma)]$ to be polynomially lower bounded.
And it is what we observe in section \ref{sec:numerical_analysis}.

For the second case in Fig.~\ref{fig:qaoa_gadget_2},
we can again use sympauli to study the Pauli paths.
We can order the gadgets arbitrarily because they commute,
so let's again choose an order that ends on an edge $e$.
We enumerate all possible $4^3$ Pauli strings on
$\mathcal{H}_e\otimes\mathcal{H}_u\otimes\mathcal{H}_v$
and evolve them through the gadget and only observe
a purely ancilla supported term when the Pauli
string is of the form $P_e\otimes P_u\otimes P_v$, $P_u,P_v\in\{I,Y\}$.
And note that such a Pauli term shouldn't have
support on any other vertex qubit for it to be counted in $F(\sigma)$.
Since the Hamiltonian only has $I_e\otimes Z_u \otimes Z_v$ terms,
and evolving them through the gadget once doesn't give us terms with
$Y_u$ or $Y_v$ (which can be verified using sympauli),
we don't get terms of this form towards the end of the circuit.
However we note that it becomes difficult to analyze the propagation
of these terms as the relevant subgraph becomes larger.
\end{widetext}

\end{document}